\documentclass[fleqn]{article}
\usepackage{multicol}
\usepackage{graphicx} 
\usepackage{subfig}
\usepackage{geometry}
\usepackage[utf8]{inputenc}
\usepackage[T1]{fontenc}
\usepackage{setspace}
\usepackage{amsmath}
\usepackage{bbm}
\usepackage{float}
\usepackage{microtype}
\usepackage{amssymb}
\usepackage{amsthm}
\newtheorem{mydef}{Definition}
\usepackage{tikz}
\usepackage{mathtools}
\usepackage{enumerate}
\usepackage{authblk}
\usepackage{hyperref}
\usepackage{xcolor}
\usepackage[
    backend=biber,
    style=phys]{biblatex} 
\newcommand{\denop}{\mathcal{D}}
\newcommand{\hs}{\mathcal{H}}

\newcommand{\id}{\mathbbm{1}}

\DeclareMathOperator{\Tr}{Tr}
\newtheorem{theorem}{Theorem}

\newtheorem{prop}[theorem]{Proposition}
\newtheorem{algorithm}[theorem]{Algorithm}

\title{Computation of the Smooth Max-Mutual Information via Semidefinite Programming}
\author[1]{Christopher Popp}
\author[2]{Tobias C. Sutter}
\author[3]{Beatrix C. Hiesmayr}
\affil[1,2,3]{University of Vienna, Faculty of Physics, Währingerstrasse 17, 1090 Vienna.\vspace{3.5mm}}
\affil[1]{christopher.popp@univie.ac.at}
\affil[2]{tobias.christoph.sutter@univie.ac.at}
\affil[3]{beatrix.hiesmayr@univie.ac.at}

\date{}

\begin{document}

\onehalfspacing
\maketitle

\begin{abstract}
We present an iterative algorithm based on semidefinite programming (SDP) for computing the quantum smooth max-mutual information $I^\varepsilon_{\max}(\rho_{AB})$ of bipartite quantum states in any dimension. The algorithm is accurate if a rank condition for marginal states within the smoothing environment is satisfied and provides an upper bound otherwise. Central to our method is a novel SDP, for which we establish primal and dual formulations and prove strong duality. With the direct application of bounding the one-shot distillable key of a quantum state, this contribution extends SDP-based techniques in quantum information theory. Thereby it improves the capabilities to compute or estimate information measures with application to various quantum information processing tasks.
\end{abstract}

\section{Introduction}
Quantum information theory has extended our understanding of information processing, communication, and cryptography by leveraging the principles of quantum mechanics \cite{wilde_quantum_2013}. Central to this field are entropic quantities that quantify correlations, uncertainties, and operational capabilities of quantum systems (cf. \cite{khatri_principles_2024} for a comprehensive overview). In particular, considering small deviations from the quantum states at hand, smoothed versions of entropic and information measures, such as the smooth max-mutual information $I^\varepsilon_{max}(\rho_{AB})$ \cite{buscemi_quantum_2010, datta_min-_2009} for bipartite quantum states $\rho_{AB}$ are essential for characterizing one-shot operational tasks (cf. Ref.~\cite{tomamichel_quantum_2016}) in noisy environments with finite resources.
\\
The smooth max-mutual information captures the correlation between subsystems $A$ and $B$, accounting for small errors due to noise or approximations due to limited resources. This quantity is particularly relevant in cryptographic settings, where it can be combined with other measures like the hypothesis testing mutual information \cite{hiai_proper_1991} to bound the one-shot distillable key \cite{khatri_second-order_2021}. This provides practical estimates for secure key rates achievable from quantum states \cite{khatri_principles_2024}. However, computing $I^\varepsilon_{max}$ exactly remains a computationally challenging optimization problem.
Existing methods in quantum information theory have increasingly relied on semidefinite programming (SDP) to compute, approximate, or bound entropic quantities efficiently \cite{skrzypczyk_semidefinite_2023,  tavakoli_semidefinite_2024,watrous_semidefinite_2009, watrous_simpler_2012, dupuis_generalized_2012}. For instance, SDP formulations have been successfully applied to compute the hypothesis testing mutual information and conditional smooth min- and max-entropies \cite{nuradha_fidelity-based_2024}. 
\\
Despite these advances, an SDP-based approach for the smooth max-mutual information is lacking.
In Ref.~\cite{nuradha_fidelity-based_2024}, the authors recently used a so-called mountain-climbing-type algorithm \cite{konno_cutting_1976} to use SDPs iteratively to calculate the smooth min-entropy, an information quantity not allowing a direct formulation as an SDP. In this work, we employ a similar approach by introducing an SDP-based iterative algorithm for computing $I^\varepsilon_{max}(\rho_{AB})$. Crucially, the quantity considered in Ref.~\cite{nuradha_fidelity-based_2024} and $I^\varepsilon_{max}$ differ in their smoothing structure, requiring an alternative approach, including a novel SDP at the core of the algorithm. We characterize this SDP by its prime-dual formulations, show strong duality, and prove that the iterative algorithm computes the smooth max-mutual information accurately, assuming a rank constraint regarding the marginal states in the smoothing neighborhood.
\\
Section \ref{sec:background} establishes the notation and the information quantities and SDP methods applied in this work. In Section \ref{sec:sptapprox}, we present an SDP-based algorithm to compute the smooth max-mutual information. We characterize the underlying SDP by its primal and dual form and prove the accuracy of the algorithm under the assumption of a rank constraint. Finally, we conclude our results in Section \ref{sec:conclusion}. 
\section{Background} \label{sec:background}
\subsection{Notation} \label{sec:notation}
Consider the bipartite Hilbert space $\hs = \hs_A \otimes \hs_B$ of two parties, $A$ and $B$. The set of linear operators on $\hs$ is denoted by $L(\hs)$. The set of positive semidefinite operators $X\geq0$ is denoted as $L_+(\hs)$ and the set of density operators, or quantum states, satisfying $\rho \in L_+(\hs)$ and $\Tr(\rho)=1$ is written as $\denop(\hs)$. The marginal states of a multipartite state are denoted as $\rho_{A/B}:=\Tr_{B/A}[\rho_{AB}]$. If $X$ is positive definite, we denote it by $X > 0$. A superoperator $\Phi$ is a linear operator mapping $L(\hs) \rightarrow L(\hs)$. $\Phi$ is hermiticity preserving, if $\Phi(X)$ is hermitian for every hermitian $X$. The adjoint superoperator $\Phi^\dagger$ is defined by satisfying $ \Tr[\Phi(X)Y] = \Tr[X\Phi^\dagger(Y)]$ for all $X,Y \in L(\hs)$.
\subsection{The max-relative entropy and the smooth max-mutual information}
The information measure analyzed in this work involves an optimization over an environment around the quantum state of interest, called smoothing. For that purpose, we use the so-called sine distance \cite{rastegin_sine_2006} based on the fidelity \cite{uhlmann_transition_1976}:
\begin{mydef}[Fidelity, sine distance]\label{def:fidelitySineDistance} \ \\
 The fidelity of two quantum states $\rho$ and $\sigma$ is defined as follows:
 \begin{flalign}
     \mathcal{F}(\rho, \sigma) := \left(\Tr[\sqrt{\sqrt{\sigma} \rho \sqrt{\sigma}}]\right)^2
 \end{flalign}
 The sine distance is defined as:
 \begin{flalign}
     P(\rho,\sigma) := \sqrt{1-\mathcal{F}(\rho,\sigma)}     
 \end{flalign}
\end{mydef} \noindent
Using this distance measure, the smoothing environment is defined as a ball around the quantum state for which the information measure is to be evaluated.
\begin{mydef}[Smoothing environment] \label{def:smoothinBall} \ \\
Using the sine distance, the smoothing environment $B^{\varepsilon}(\rho)$ around a state $\rho$ is defined as
    \begin{flalign}
        \label{eq:smoothBall}
        B^\varepsilon(\rho) := \left\{ \tilde{\rho} \in \denop(\hs): P(\rho,\tilde{\rho}) \leq \varepsilon \right\} 
        = \left\{ \tilde{\rho} : \mathcal{F}(\rho, \tilde{\rho}) \geq 1-\varepsilon^2 \right\}.
    \end{flalign}
\end{mydef} \noindent
The quantity of interest in this work is related to the max-relative entropy $D_{max}$ \cite{datta_min-_2009}. Here, we analyze a smoothed generalized mutual information $I^\varepsilon_{max}$ \cite{buscemi_quantum_2010} based on $D_{max}$, which can be used for quantifying the capability of a quantum state to be used for the task of secret-key distillation (cf., e.g., \cite{khatri_second-order_2021,khatri_principles_2024}).
\begin{mydef}[Max-relative entropy, smooth max-mutual information] 
\label{def:maxrelEntInf}
    \begin{flalign}
        \label{eq:smoothMaxMutualInf}
        D_{max}(\rho||\sigma) &:= \begin{cases}
            \log_2 \left\| \sigma^{-\frac{1}{2}} \rho \sigma^{-\frac{1}{2}}  \right\|_{\infty} ~~~\text{ if supp}(\rho) \subseteq \text{supp}(\sigma) \\
            + \infty ~~~~~~~~~~~~~~~~~~~~~ \text{ else}
        \end{cases} \\
        I_{max}^\varepsilon(\rho_{AB}) &:= \inf_{\tilde{\rho}_{AB} \in B^\varepsilon(\rho_{AB})} D_{max}(\tilde{\rho}_{AB} || \rho_A \otimes \tilde{\rho}_B)
    \end{flalign}
    where $B^\varepsilon(\rho_{AB})$ is as in Definition~\ref{def:smoothinBall}, $|| \cdot||_\infty$ is the spectral norm and $\tilde{\rho}_B := \Tr_A[\tilde{\rho}_{AB}]$.
\end{mydef}
\subsection{Semidefinite Programming}
Semidefinite programs are constrained optimization problems with positive semidefinite optimization variables that often occur in quantum information theory.
\begin{mydef}[Semidefinite program] \label{def:sdp}\ \\
    Let $\Phi$ be a hermiticity preserving superoperator and $A$ and $B$ be hermitian operators. A semidefinite program (SDP) corresponds to the following two optimization problems over positive semidefinite operators $X,Y$. \\
    The primal SDP:
    \begin{flalign} \label{eq:primal}
        &\text{maximize } \Tr[AX] \\
        &\text{subject to }~ \Phi(X) \leq B, ~ X\geq0. \notag
    \end{flalign}
    The dual SDP:
    \begin{flalign} \label{eq:dual}
        &\text{minimize } \Tr[BY] \\
        &\text{subject to }~ \Phi^\dagger(Y) \geq A, ~ Y\geq0, \notag
    \end{flalign}
    where $\Phi^\dagger$ is the adjoint superoperator of $\Phi$.
\end{mydef} \noindent
A variable $X$ satisfying the constraints of the primal SDP is called a feasible point. If $X$ satisfies $\Phi(X)<B$ and $X > 0$, it is called a strictly feasible point. Similarly, $Y$ is feasible if it satisfies the constraints of the dual SDP and strictly feasible if the constraints are strictly satisfied. An SDP is said to satisfy the strong duality property if the optimal values of \eqref{eq:primal} and \eqref{eq:dual} of the primal and dual problems are equal.
\section{Computing the smooth max-mutual information with an SDP-based algorithm}
\label{sec:sptapprox}
The max-relative entropy $D_{max}(\rho||\sigma)$ of two quantum states $\rho$ and $\sigma$ can be characterized by the following optimization problem (cf. Ref.~\cite{khatri_principles_2024}):
\begin{flalign}
    \label{eq:opt_dmax}
    D_{max}(\rho||\sigma) = \log_2 \inf_{\lambda \geq 0} \lbrace \lambda:  \lambda \sigma - \rho \geq 0 \rbrace.
\end{flalign}
The inner optimization can be formulated as an SDP in its dual form by setting $Y\equiv\lambda,~A=\rho, ~B = 1, ~\Phi^\dagger(Y)=Y\sigma$ in Definition \ref{def:sdp}. The same does not hold for the smooth max-mutual information $I^\varepsilon_{max}(\rho_{AB})$. This can be seen by writing $I_{max}^\varepsilon$ using the characterization for the max-relative entropy \eqref{eq:opt_dmax}, implying
\begin{flalign}
    \label{eq:smoothMaxMutAlt}
    I_{max}^\varepsilon(\rho_{AB}) := \log_2 \left( \inf_{\tilde{\rho}_{AB} \in B^\varepsilon(\rho_{AB})} 
    \inf_{\lambda \geq 0} \left\{
        \lambda: \lambda \rho_A \otimes \tilde{\rho}_B - \tilde{\rho}_{AB} \geq 0
    \right\}
    \right) := \log_2(\lambda_{min}).
\end{flalign}
As the constraint involves a bilinear term in the variables $\tilde{\rho}_{AB}$ and $\lambda$, this cannot be directly solved by a SDP. We propose an iterative approach applicable for states $\rho_{AB}$ and $\varepsilon > 0$ for which $B^\varepsilon(\rho_{AB})$ contains only full-rank marginal states $\rho_A := \Tr_B[\rho_{AB}]$ and $\tilde{\rho}_B:=\Tr_A[\tilde{\rho}_{AB}]$. Alternately fixing one of the variables and solving an SDP for the other, this seesaw or mountain-climbing algorithm converges to the optimal solution after a finite amount of iterations \cite{konno_cutting_1976}, providing an upper bound approximating $I^{\varepsilon}_{max}(\rho_{AB})$.
Define the $i$-th iteration of the procedure as follows:
\begin{algorithm} \label{thm:algorithm}\ \\
\begin{enumerate}
    \item If $i=1$, choose $\rho_{AB}^i = \rho_{AB}$.
    \item Solve the following SDP, with the solution defining 
    \begin{flalign}
    \label{eq:sdpStep2}
        \lambda^i := \inf_{\lambda \geq 0} 
        \left\{ 
            \lambda : \lambda \rho_A \otimes \rho_B^i - \rho_{AB}^i \geq 0
        \right\}.
    \end{flalign}
    \item Solve the following SDP, with the solution defining
    \begin{flalign}
        \label{eq:sdpStep3}
        \mu^i := \sup_{\substack{\mu \geq 0, \\ \tilde{\rho}_{AB}\in B^\varepsilon(\rho_{AB})}}
        \left\{
            \mu: \lambda^i \rho_A \otimes \tilde{\rho}_B - \tilde{\rho}_{AB} - \mu \mathbbm{1} 
            \geq 0
        \right\}
    \end{flalign}
    \item If $\mu^i = 0$ stop and set $\lambda_{min} = \lambda^i$. Else define $\rho_{AB}^{i+1}$ to be state $\tilde{\rho}_{AB} \in B^\varepsilon(\rho_{AB})$ maximizing the expression in step 3 and use it in step 2 for the next iteration.
\end{enumerate}
\end{algorithm} \noindent
Note that for all iterations, $B^\varepsilon(\rho_{AB})$ remains fixed around $\rho_{AB}$. We now characterize the SDP of \eqref{eq:sdpStep3} by its primal and dual form and show strong duality.
\begin{prop}
    \label{thm:strongduality}
    The primal and dual SDPs of 
    \begin{flalign}
    \label{eq:sdpStep2Prop}
    \sup_{\substack{\mu \geq 0, \\ \tilde{\rho}_{AB}\in B^\varepsilon(\rho_{AB})}}
        \left\{
            \mu: \lambda^i \rho_A \otimes \tilde{\rho}_B - \tilde{\rho}_{AB} - \mu \mathbbm{1} 
            \geq 0
        \right\}
    \end{flalign}
    can be written in the following forms:
    \begin{flalign}
        &\sup_{\substack{\mu \geq 0, \\ \tilde{\rho}_{AB}\in \denop(\hs), \\ X \in L(\hs)}}
    \left\{
        \mu : \lambda \rho_A \otimes \tilde{\rho}_B - \tilde{\rho}_{AB} - \mu \mathbbm{1} \geq 0,
        ~\mathrm{Re} \lbrace \Tr[X]\rbrace \geq \sqrt{1-\varepsilon^2},  
        ~\begin{pmatrix}
            \rho_{AB} & X \\
            X^\dagger & \tilde{\rho}_{AB}
        \end{pmatrix}
        \geq 0
    \right\} ~~~~~~~~\text{(primal)}, \\
    &        \inf_{\substack{W, D \geq 0\\ \eta \in \mathbbm{R}, \nu \geq 0}} \left\{ 
        \eta - 2\nu \sqrt{1-\varepsilon^2} + \Tr[\rho D]: ~
        \Tr[W] \geq 1,     ~
         \begin{pmatrix}  D & \nu\mathbbm{1} \\ 
         \nu \mathbbm{1} & \eta \mathbbm{1}- (\lambda M -W)\end{pmatrix} \geq 0 
    \right\} ~~~~~~~~~~~~~~~~~\text{(dual)},
    \end{flalign}
    and strong duality holds.
\end{prop}
\begin{proof}    
First note that by \eqref{eq:smoothBall}, the $\varepsilon$-ball based on the sine distance can be characterized by the root fidelity, for which we can use the characterization 
$
    \sqrt{\mathcal{F}}(\rho,\tilde{\rho}) := \sqrt{\mathcal{F}(\rho,\tilde{\rho})}= \sup_{X \in L(\hs)} 
        \left\{ \mathrm{Re} \lbrace \Tr[X] \rbrace:  
            \begin{pmatrix}
            \rho & X \\
            X^\dagger & \tilde{\rho}
            \end{pmatrix}
            \geq 0
 \right\}
$
\cite{watrous_simpler_2012} to write the optimization problem in \eqref{eq:sdpStep2Prop} as:
\begin{flalign}
    \label{eq:sdpStep2_1}
    \sup_{\substack{\mu \geq 0, \\ \tilde{\rho}_{AB}\in \denop(\hs, \\ X \in L(\hs)}}
    \left\{
        \mu : \lambda \rho_A \otimes \tilde{\rho}_B - \tilde{\rho}_{AB} - \mu \mathbbm{1} \geq 0,
        ~\mathrm{Re} \lbrace \Tr[X]\rbrace \geq \sqrt{1-\varepsilon^2},  
        ~\begin{pmatrix}
            \rho_{AB} & X \\
            X^\dagger & \tilde{\rho}_{AB}
        \end{pmatrix}
        \geq 0
    \right\}.
\end{flalign}
For the primal problem in the standard form $
    \sup_{Z \geq 0} \left\{ \Tr[AZ] : \Phi(Z) \leq B \right\},     
$ where $A$ and $B$ are hermitian operators with suitable block sizes so that all matrix products are well defined and $\Phi$ is a hermiticity-preserving superoperater, one finds:
\begin{flalign}
    Z &= \begin{pmatrix}
        \mu & 0 & 0 \\
        0 & D & X \\
        0 & X^\dagger & \tilde{\rho}_{AB}
    \end{pmatrix}, ~~~
    A = \begin{pmatrix}
        1 & 0 & 0 \\
        0 & 0 & 0 \\
        0 & 0 & 0
    \end{pmatrix}, ~~~
        B = \begin{pmatrix}
        0 & 0 & 0 & 0 & 0 \\
        0 & 1 & 0 & 0 & 0 \\
        0 & 0 & -1 & 0 & 0 \\
        0 & 0 & 0 & -\sqrt{1-\varepsilon^2} & 0 \\
        0 & 0 & 0 & 0 &  \begin{bmatrix} \rho_{AB} & 0 \\ 0 & 0  \end{bmatrix} \\
    \end{pmatrix}
    \\
    \Phi(Z) &=
    \begin{pmatrix}
        -(\lambda \rho_A \otimes \tilde{\rho}_B - \tilde{\rho}_{AB} - \mu \mathbbm{1}) & 0 & 0 & 0 & 0\\
        0 & \Tr[\tilde{\rho}_{AB}] & 0 & 0 & 0 \\
        0 & 0 & -\Tr[\tilde{\rho}_{AB}] & 0 & 0 \\
        0 & 0 & 0 & -Re\lbrace\Tr[X]\rbrace & 0\\
        0 & 0 & 0 & 0 & \begin{bmatrix} 0 & -X \\ -X^\dagger & -\tilde{\rho}_{AB} \end{bmatrix} \\
    \end{pmatrix}.
\end{flalign}
We use this to derive the dual problem and show strong duality. The dual problem is
$\inf_{Y \geq 0} \left\{ \Tr[BY] : \Phi^\dagger(Y) \geq A  \right\}$,
where $\Tr[\Phi(Z)Y] = \Tr[Z\Phi^\dagger(Y)]$ defines the adjoint superoperator. To derive $\Phi^\dagger$, set 
\begin{flalign}
    Y = \begin{pmatrix}
        W & 0 & 0 & 0 & 0 \\
        0 & \mu_1 & 0 & 0 & 0 \\
        0 & 0 & \mu_2 & 0 & 0  \\
        0 & 0 & 0 & \nu & 0  \\
        0 & 0 & 0 & 0 & \begin{bmatrix}  D_1 & V \\ V^\dagger & D_2 \end{bmatrix}  \\        
    \end{pmatrix},
\end{flalign}
with $W \in L_+(\hs), D_1, D_2, V \in L(\hs)$ such that $Y \geq 0$ and $\Tr[BY]$ is well-defined. We then have
\begin{flalign}
    \Tr[\Phi(Z)Y] &= \Tr\left[-(\lambda \rho_A \otimes \tilde{\rho}_B - \tilde{\rho}_{AB} - \mu\mathbbm{1})W \right] + \Tr[\tilde{\rho}_{AB}](\mu_1-\mu_2) - \mathrm{Re}\lbrace \Tr[X] \rbrace \nu \notag \\
   &~~~ - \Tr \left[ 
         \begin{pmatrix} 0 & X \\ X^\dagger & \tilde{\rho}_{AB} \end{pmatrix}
         \begin{pmatrix} D_1 & V \\ V^\dagger & D_2 \end{pmatrix}
    \right]  \\
    &= \mu \Tr[W] - \Tr[(\lambda \rho_A \otimes \tilde{\rho}_B-\tilde{\rho}_{AB})W] + \Tr[(\mu_1 - \mu_2)\tilde{\rho}_{AB}] -\mathrm{Re}\lbrace \Tr[X]\rbrace \nu \notag \\
    &~~~  - \Tr[XV^\dagger + X^\dagger V + \tilde{\rho}_{AB}D_2] \\
    &= \mu \Tr[W] + \Tr[\left\{  (\mu_1-\mu_2)\mathbbm{1} - D_2 - (\lambda M - W)  \right\}\tilde{\rho}_{AB}]
    - \mathrm{Re} \lbrace \Tr[X(\nu \mathbbm{1} + 2V^\dagger)]\rbrace.
\end{flalign}
In the last equation, we have defined the operator $M$, using an ancillary space $A'$ of the same size as $A$:
\begin{flalign}
    \label{eq:Mopdef}
    M := \mathbbm{1}_A \otimes \lbrace \Tr_{A'}[W_{A'B}(\rho_{A'} \otimes \mathbbm{1}_B)] \rbrace 
    \implies \Tr_{AB}[M \tilde{\rho}_{AB}] = \Tr_{A'B}[ (\rho_{A'}\otimes \tilde{\rho}_B) W].
\end{flalign}
This implies the following three relations:
\begin{enumerate}
\item
 \begin{flalign}
    &\Tr[\Phi(Z)Y] = \mu \Tr[W] 
    + \Tr \left[ 
        \begin{pmatrix}
            0 & \frac{\nu}{2}\mathbbm{1} +V \\
            \frac{\nu}{2}\mathbbm{1} +V^\dagger & (\mu_1-\mu_2)\mathbbm{1} - D_2 -(\lambda M-W)
        \end{pmatrix}
        \begin{pmatrix}
            D & X \\
           X^\dagger & \tilde{\rho}_{AB}
        \end{pmatrix}
    \right] \\
    &\implies \Phi^{\dagger}(Y) =  \begin{pmatrix}
        \Tr[W] & 0 \\
        0 &  \begin{bmatrix}
            0 & \frac{\nu}{2}\mathbbm{1} +V \\
            \frac{\nu}{2}\mathbbm{1} +V^\dagger & (\mu_1-\mu_2)\mathbbm{1} - D_2 -(\lambda M-W)
        \end{bmatrix}
    \end{pmatrix}
\end{flalign}
\item
\begin{flalign}
    \Phi^\dagger(Y) &\geq A 
    \Leftrightarrow  
    \Tr[W] \geq 1,   \begin{pmatrix}
       0 & \frac{\nu}{2} \mathbbm{1} + V \\ \frac{\nu}{2} \mathbbm{1} + V^\dagger & (\mu_1 - \mu_2)\mathbbm{1} - D_2 - (\lambda M - W)  \end{pmatrix} \geq 0
\end{flalign}
\item
\begin{flalign}
    \Tr[BY] &= \mu_1 - \mu_2 - \nu \sqrt{1-\varepsilon^2} + \Tr[\rho D_1].
\end{flalign}
\end{enumerate}
Hence, the dual program can be written as follows:
\begin{flalign} 
    \inf_{\substack{W \geq 0\\ \mu_1, \mu_2, \nu \geq 0}} \bigg\{
        &\mu_1 - \mu_2 - \nu \sqrt{1-\varepsilon^2} + \Tr[\rho D_1]: \notag \\
        &\Tr[W] \geq 1,     
         \begin{pmatrix} 0 & \frac{\nu}{2} \mathbbm{1} + V \\ \frac{\nu}{2} \mathbbm{1} + V^\dagger & (\mu_1 - \mu_2)\mathbbm{1} - D_2 - (\lambda M - W)  \end{pmatrix} \geq 0, ~
         \begin{pmatrix}  D_1 & V \\ V^\dagger & D_2  \end{pmatrix} \geq 0 
    \bigg\}.
    \label{eq:dualproof}
\end{flalign}
We can simplify the program by defining $\eta:=(\mu_1-\mu_2) \in \mathbbm{R}$ and setting $\nu/2 \rightarrow \nu \geq 0$ in \eqref{eq:dualproof}, yielding
\begin{flalign}
    \label{eq:original_conditions}
    &\begin{pmatrix} 0 & \nu \mathbbm{1} + V \\ \nu \mathbbm{1} + V^\dagger & (\eta\mathbbm{1} - D_2 - (\lambda M - W)  \end{pmatrix} \geq 0, ~\begin{pmatrix}  D_1 & V \\ V^\dagger & D_2  \end{pmatrix} \geq 0. 
\end{flalign}
We can simplify this further by noting that the first condition requires $V = -\nu \id$, which implies $D_2 \leq \eta \id - (\lambda M -W)$. Setting $D_2=\eta \id -(\lambda M -W)$ and $V = -\nu\id$, one finds with the second constraint
\begin{flalign}
    \label{eq:new_conditions}
 \begin{pmatrix}  D_1 & \nu\id \\ \nu\id & \eta -(\lambda M -W) \end{pmatrix}\geq 0,
\end{flalign}
which holds if and only if \eqref{eq:original_conditions} holds. For invertible $D_1$ or $\nu=0$, this is implied by the Schur complement, and in case $D_1$ is singular and $\nu>0$, neither \eqref{eq:original_conditions} nor \eqref{eq:new_conditions} can be fulfilled. Hence, we can write the final program using \eqref{eq:new_conditions} and renaming $D_1 \rightarrow D$ as:
\begin{flalign}
        \inf_{\substack{W, D \geq 0\\ \eta \in \mathbbm{R}, \nu \geq 0}} \left\{ 
        \eta - 2\nu \sqrt{1-\varepsilon^2} + \Tr[\rho D]: ~
        \Tr[W] \geq 1,     ~
         \begin{pmatrix}  D & \nu\mathbbm{1} \\ 
         \nu \mathbbm{1} & \eta \mathbbm{1}- (\lambda M -W)\end{pmatrix} \geq 0 
    \right\}.
\end{flalign}
Finally, we conclude strong duality for the primal and dual problem. Set $\tilde{\rho}_{AB} = \rho_{AB}$, $\mu$ such that $\lambda \rho_A \otimes \rho_B - \rho_{AB} - \mu \mathbbm{1} \geq 0$, $X = \sqrt{1-\varepsilon^2}$, implying $\begin{pmatrix}
    \rho_{AB} & X \\
    X^\dagger & \tilde{\rho}_{AB}
\end{pmatrix} \geq 0$ and thus being a feasible point of the primal SDP. Let $\alpha > 1,~\nu >0,~\beta >0, ~\eta = \beta + \nu + (\lambda -1)\alpha, ~D=(\nu +\beta) \mathbbm{1}$ and $W = \alpha \mathbbm{1}$. Then, $\Tr[W] > 1,~ W > 0,~D>0$ and by \eqref{eq:Mopdef} $M = \alpha \mathbbm{1}$ and $\begin{pmatrix}
    D & \nu \mathbbm{1} \\
    \nu \mathbbm{1} & \eta \mathbbm{1} - (\lambda M - W)
\end{pmatrix} 
= 
\begin{pmatrix}
    (\nu+\beta) \mathbbm{1} & \nu \mathbbm{1} \\
    \nu \mathbbm{1} & (\nu+\beta) \mathbbm{1}
\end{pmatrix} > 0$. Consequently, this is a strictly feasible point and strong duality is implied by Slater's condition (cf. \cite{boyd_convex_2006}).
\end{proof}
\begin{prop}
    \label{thm:approxProgram}
    If for $\rho_{AB}$ and all $\tilde{\rho}_{AB} \in B^{\varepsilon}(\rho_{AB})$, $\rho_A\otimes\tilde{\rho}_B $ is positive definite, the procedure above converges to $\lambda_{min}$ of \eqref{eq:smoothMaxMutAlt}.
\end{prop}
\begin{proof}
    Let $\lambda^i$ be as in \eqref{eq:sdpStep2}. Assume that $\lambda^i > \lambda_{min}$ as in \eqref{eq:smoothMaxMutAlt}. Then, there exists a $\hat{\rho}_{AB}$ and $\hat{\lambda} = \lambda^i - \hat{\delta}$ with $\hat{\delta} >0$ and $\hat{\lambda}\rho_A \otimes \hat{\rho}_B - \hat{\rho}_{AB} \geq 0$. This implies $\lambda^i \rho_A \otimes \hat{\rho}_B - \hat{\rho}_{AB} - \hat{\delta}\rho_A \otimes \hat{\rho}_B \geq 0$ and consequently, because $\hat{\delta}> 0$ and by assumption $\rho_A \otimes \hat{\rho}_B > 0$, we must have $\lambda^i \rho_A \otimes \hat{\rho}_B - \hat{\rho}_{AB} >0$. This indicates that there exist $\hat{\rho}_{AB} \in B^{\varepsilon}(\rho_{AB})$ and $\hat{\mu}>0$ such that $\lambda^i \rho_A \otimes \hat{\rho}_B - \hat{\rho}_{AB} - \hat{\mu}\id \geq 0$. This implies that the SDP \eqref{eq:sdpStep3} has a feasible solution $(\mu^i, \rho^{i+1}_{AB})$ and consequently $\lambda^i \rho_A \otimes \rho^{i+1}_B - \rho^{i+1}_{AB} > 0$. Thus, using $\rho^{i+1}_{AB}$ for the next iteration in \eqref{eq:sdpStep2}, will return a $\lambda^{i+1} < \lambda^i$. In summary, we have shown that if $\lambda^i > \lambda^{min}$, then there exists a feasible solution $\mu^i$ of the SDP \eqref{eq:sdpStep3}, provided that all states have full rank. This solution implies a feasible improved solution $\lambda^{i+1}<\lambda^i$ of \eqref{eq:sdpStep2} for the next iteration. Due to the strong duality of the SDP involved (see Proposition~\ref{thm:strongduality}), if it exists, a feasible solution will be found by the program. Failing to do so in step 3 in Algorithm \ref{thm:algorithm} consequently implies that no feasible $\mu^i > 0$ exists and therefore $\lambda^i = \lambda_{min}$. Proposition 2.3. in \cite{konno_cutting_1976} implies that this seesaw/mountain-climbing-type algorithm converges to the optimal solution after a finite number of steps.
\end{proof} \noindent
Before concluding, we briefly discuss the role of the rank constraint $\rho_A \otimes \hat{\rho}_B > 0$: If this constraint is violated, then $\lambda^i > \lambda_{min}$ does not necessarily imply $\lambda^i \rho_A \otimes \hat{\rho}_B - \hat{\rho}_{AB} >0$. Instead, $\lambda^i \rho_A \otimes \hat{\rho}_B - \hat{\rho}_{AB}$ and $\rho_A \otimes \hat{\rho}_B$ may have nontrivial intersecting kernels, such that $\lambda^i \rho_A \otimes \hat{\rho}_B - \hat{\rho}_{AB} - \hat{\delta}\rho_A \otimes \hat{\rho}_B \geq 0$ holds for some $\hat{\delta}>0$, but $\lambda^i \rho_A \otimes \hat{\rho}_B - \hat{\rho}_{AB} - \mu\id < 0$ for all $\mu > 0$. In this case, the SDP in \eqref{eq:sdpStep3} results in $\mu=0$ and the algorithm stops without necessarily having found the optimal solution $\lambda_{min}$.
\section{Conclusion}
\label{sec:conclusion}
In this work, we developed an SDP-based iterative algorithm to compute the smooth max-mutual information $I^\varepsilon_{\max}(\rho_{AB})$ of bipartite quantum states. The algorithm was proven to be accurate under the assumption that for all states $\tilde{\rho}_{AB}$ in the smoothing environment around $\rho_{AB}$, the marginal states $\rho_A \otimes \tilde{\rho}_B$ have full rank. 
\\
A key contribution of this work is the characterization of a previously unexplored semidefinite program, which forms a core component of the iterative procedure. We established its primal and dual formulations and proved strong duality, thereby laying a rigorous foundation for its use in quantum information theory.
Our results extend the landscape of SDP-based techniques for evaluating operational quantities, complementing existing methods as used for the computation of, e.g., the hypothesis testing mutual information \cite{dupuis_generalized_2012, datta_second-order_2016}, or the conditional smooth min- and max-entropy \cite{nuradha_fidelity-based_2024}. Notably, since the smooth max-mutual information can be combined with the hypothesis testing mutual information to lower-bound the one-shot distillable key (cf. Ref.~ \cite{khatri_principles_2024}), our algorithm offers a tool for estimating cryptographic capacities of quantum states that satisfy the rank constraints. 
\\
If the constraints are not fulfilled, the algorithm may not yield optimal solutions. Future work may focus on relaxing the rank assumptions, which would broaden the applicability of our method to general quantum states. In Ref.~\cite{popp_local_2025}, it is shown that the smooth max-mutual information of $\rho_{AB}$ is invariant under local unitary transformations $U_A$ acting on $A$ and general local isometric transformations $V_{B' \rightarrow B}$, mapping a smaller system $B'$ to $B$. If for a state $\rho_{AB}$ the rank condition is violated, but there exists a state $\hat{\rho}_{AB'}$ with $\rho_{AB} = U_A\otimes V_{B' \rightarrow B}~\hat{\rho}_{AB'}~U_A^\dagger \otimes V_{B' \rightarrow B}^\dagger$ that satisfies the condition, then the smooth max-mutual information can be computed with the presented algorithm for $\hat{\rho}_{AB'}$, instead. One potential way to lift the rank conditions is to analyze for which states such a pre-image under local transformations exists.
Moreover, the fully characterized SDP presented in Proposition~\ref{thm:strongduality} may find utility beyond this context, potentially serving as a building block for evaluating other information-theoretic quantities or optimization problems in quantum theory and related fields.
\printbibliography

\section*{Author Contributions} Conceptualization, C.P.; validation, C.P., T.C.S. and  B.C.H.; formal analysis, C.P.; writing---original draft preparation, C.P.; writing---review and editing, C.P., T.C.S. and B.C.H.   All authors have read and agreed to the published version of the manuscript.

\section*{Acknowledgments} This research was funded in whole, or in part, by the Austrian Science Fund (FWF) [10.55776/P36102]. For the purpose of open access, the author has applied a CC BY public copyright license to any Author Accepted Manuscript version arising from this submission.

\end{document}